\documentclass[12pt,book,reqno]{amsart}
\usepackage[T1]{fontenc}
\usepackage[latin9]{inputenc}
\usepackage{amssymb}
\usepackage{xcolor}
\usepackage{graphicx,color}
\usepackage{amsmath, amssymb, graphics}
\usepackage{qcircuit}
\usepackage{braket}

\makeatletter
\numberwithin{equation}{section} %% Comment out for sequentially-numbered
\numberwithin{figure}{section} %% Comment out for sequentially-numbered
  \@ifundefined{theoremstyle}{\usepackage{amsthm}}{}
  \theoremstyle{plain}
  \newtheorem{thm}{Theorem}[section]
  \theoremstyle{plain}
  
  \theoremstyle{plain}
  \newtheorem{prop}[thm]{Proposition}
  \theoremstyle{remark}
  \newtheorem{rem}[thm]{Remark}
  \theoremstyle{remark}
  
  \theoremstyle{plain}
  \newtheorem{lem}[thm]{Lemma}
  \newtheorem{mydef}{Definition}

\usepackage{graphicx,xcolor,amsmath,multicol}
\usepackage{tikz}
\usepackage[english]{babel}
\usepackage{pdfpages}
\usepackage{fancyhdr}

\usepackage[colorlinks=true,linkcolor=asparagus]{hyperref}
\usepackage{times}

\definecolor{asparagus}{rgb}{0.53,0.66, 0.42}
\definecolor{xanadu}{rgb}{0.45, 0.53, 0.47}
\definecolor{alizarin}{rgb}{0.82, 0.1, 0.26}
\definecolor{ao}{rgb}{0.0, 0.5, 0.0}
\definecolor{azure}{rgb}{0.0, 0.5, 1.0}
\definecolor{awesome}{rgb}{1.0, 0.13, 0.32}
\definecolor{greenish}{rgb}{0.1, 0.5, 0.1}

\newcommand{\dd}[1]{\, \displaystyle{#1}}

\newcommand{\deter}{{\rm det}}

\newcommand{\columntwo}[2]{\begin{pmatrix} #1 \\  #2   \end{pmatrix} }

\newcommand{\matrixtwo}[4]{\begin{pmatrix} #1 &  #2 \\ #3  &  #4  \end{pmatrix} }

\newcommand{\im}[1]{\, {\rm Im}\,#1 \,}

\smallskip
\def\<{{\langle }}
\def\>{{\rangle }}

\def\ket#1{|#1\rangle}

\smallskip
\def\<{{\langle }}
\def\>{{\rangle }}

\begin{document}

\title[3-qubit state preparation]{ Preparation of 3-qubit states }

\author{  Oscar Perdomo, Nelson Castaneda and Roger Vogeler }
\date{\today}
\begin{abstract}
We will call a pure qubit state real if all its amplitudes are real numbers. We show that any real 3-qubit state can be prepared using  $R_y(\theta)$  gates and at most four controlled-$Z$ gates, and we conjecture that four is optimal. We also present an algorithm---different from the 2008 algorithm given by Znidaric, Giraud and Georgeot---that prepares any 3-qubit state using local gates and at most three controlled-$Z$ gates.  Videos showing how our method works for two- and three-qubit states can be found at  \url{https://youtu.be/LIdYSs-rE-o} and \url{https://youtu.be/Kne0Vq7gyzQ}.

\end{abstract}

%\subjclass[2000]{53C42, 53A10}
%\maketitle

%\begin{abstract}
%We study some conditions for the helicopter Bell 525 Relentless at the speed 337.33 feet per second, the velocity at the moment it crash.

%under certain helicopter blade design conditions
%\end{abstract}

%\subjclass[2000]{53C42, 53A10}
\maketitle

\section{Introduction}

A natural problem in quantum information is the one of creating a quantum circuit---a sequence of quantum gates---that prepares any pure quantum state. The gates considered in this paper are the local gates and the controlled-$Z$ gates. There are some interesting aspects of the collection of pure qubit states with real amplitudes, which we call  {\it real}  states.  Local gates represented by matrices having all real entries are called {\it real local gates}.

For 2-qubit states,  the distinction of real states is only superficial, since for any state $\ket{\phi}$ there is a local gate $U$ such that $U\ket{\phi}$ has real amplitudes.

For 3-qubit states the situation is different. There exist 3-qubit states $\ket{\phi}$ such that  $U\ket{\phi}$ is not real for any local gate $U$. There also exist pairs of  real 3-qubit states $\ket{\phi_1}$ and $\ket{\phi_2}$ that are equivalent but not {\it real} equivalent; that is, $\ket{\phi_2}=U\ket{\phi_1}$ for some local gate $U$ but not for any real local gate, see \cite{P1}. A canonical representation for 3-qubit states was given in \cite{AA} where it is shown that for any state $\ket{\phi}$ there is a local gate $U$ such that $U\ket{\phi}$ has the  form $\lambda_1\ket{000}+\lambda_2e^{i\theta}\ket{100}+\lambda_3\ket{101}+\lambda_4\ket{110}+\lambda_5\ket{111}$, where $\theta$ and  $\lambda_i$ are all real numbers. For real states, it is known that for every real 3-qubit state $\ket{\phi}$ there is a real local  gate $U$ such that $U\ket{\phi}$ has the  form $\lambda_1\ket{000}+\lambda_2\ket{011}+\lambda_3\ket{101}+\lambda_4\ket{110}+\lambda_5\ket{111}$, see \cite{P2}. 

Any 3-qubit state can be prepared using local gates and at most three controlled-$Z$ gates, see \cite{Z}. This result motivates the following question: Can every real 3-qubit state be prepared using only real local gates and at most three controlled-$Z$ gates? In this paper we provide an algorithm that prepares any real  3-qubit state using real local gates and at most four controlled-$Z$ gates. We strongly suspect that this result is optimal, which would mean that there exist real 3-qubit states that cannot be prepared using only local real gates and fewer than four controlled-$Z$ gates.  

Our algorithm has two special features. First, it is algebraically simple, requiring neither inverse trigonometric functions nor eigenvalues. It involves only the four basic operations and square roots. Second, it extends easily to give a simple algorithm for preparing general 3-qubit states using local gates and at most three controlled-$Z$ gates.

For a general 3-qubit state $\ket{\phi}$ we show explicitly how to find local gates $W_1,\dots , W_8$ for the following circuit

$$
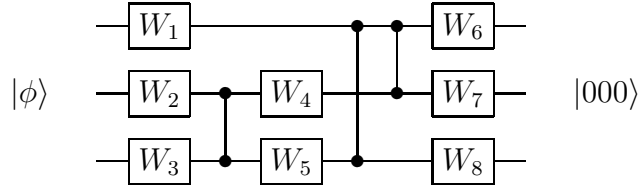
\Qcircuit @C=1em @R=.7em { 
                   &  & \gate{W_1} & \qw       & \qw            &\ctrl{2}  & \ctrl{1}  &\gate{W_6} & \qw& \\ 
\lstick{\ket{\phi}}& & \gate{W_2} & \ctrl{1}  & \gate{W_4}  & \qw   &  \ctrl{-1}  &\gate{W_7} & \qw&\rstick{\ket{000}} \\
               &       & \gate{W_3} & \ctrl{-1} & \gate{W_5}   &\ctrl{-2}  & \qw  &\gate{W_8} & \qw &
                      }$$
                    
where the indicated controlled-$Z$ gates can be removed in certain cases, as will be made clear in our explanation.

For a real 3-qubit state 
$$\ket{\phi}=w_0\ket{000}+w_1\ket{001}+w_2\ket{010}+w_3\ket{011}+w_4\ket{100}+w_5\ket{101}+w_6\ket{110}+w_7\ket{111}$$
let us define
$$\Delta(\ket{\phi})=(w_0 w_7-w_1 w_6-w_2 w_5+w_3 w_4)^2-4 (w_1 w_2-w_0 w_3) (w_5 w_6-w_4 w_7)$$

We show that if $\Delta(\ket{\phi})\ge 0$, then the circuit with three controlled-$Z$ gates shown above prepares the state $\ket{\phi}$ in ``real fashion'' because all $W_i$ reduce to real local gates. When $\Delta(\ket{\phi})<0$, then our method requires a fourth controlled-$Z$ and we find real local gates $W_0,W_2,\dots, W_8$ for the following circuit

$$
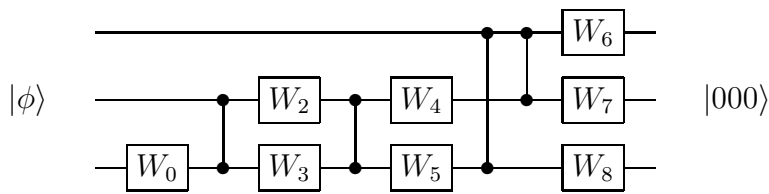
\Qcircuit @C=1em @R=.7em { 
                         & &   \qw            &  \qw &  \qw & \qw       & \qw            &\ctrl{2}  & \ctrl{1}  &\gate{W_6} & \qw& \\ 
\lstick{\ket{\phi}}& &  \qw            &\ctrl{1}     &\gate{W_2} & \ctrl{1}  & \gate{W_4}  & \qw   &  \ctrl{-1}  &\gate{W_7} & \qw&\rstick{\ket{000}} \\
                         & &   \gate{W_0} &\ctrl{-1}    &\gate{W_3} & \ctrl{-1} & \gate{W_5}   &\ctrl{-2}  & \qw  &\gate{W_8} & \qw &
                      }$$

Our primary interest was the preparation of real $3-$qubits but we were able to adapt the method we developed   for the preparation of general $3-$qubits. For simplicity we shall explain first the algorithm for the general 3-qubit state

The justification for the algorithms is based on standard linear algebra of 2-by-2 matrices. A key idea in our study is borrowed from the paper \cite{AA} where the authors introduce two 2-by-2 matrices for any 3-qubit state. We denote these matrices $A$ and $B$, defined as follows:

\begin{eqnarray*}
\xi_0\ket{000}+\xi_1\ket{001}+\dots+\xi_7\ket{111} \equiv 
\ket{0}\, \begin{pmatrix} \xi_0 &\xi_1 \\ \xi_2 & \xi_3 \end{pmatrix} + \ket{1}\, \begin{pmatrix} \xi_4 &\xi_5 \\ \xi_6 & \xi_7 \end{pmatrix}  =\ket{0}A+\ket{1}B
\end{eqnarray*}

In section 2 we provide a collection of basic linear algebra lemmas. These explain, for $\ket{\phi}=\ket{0}A+\ket{1}B$, how the matrices $A$ and $B$ change when a gate $M$ is applied to $\ket{\phi}$ and how to select a local gate $M$ so that the resulting matrices for $M\ket{\phi}$ have some desired property. Section 3 uses  linear algebra  to produce  explicit formulas for the matrices $W_i$ in the circuits. % and also provides some concrete examples.

%%
%% LINEAR ALGEBRA SECTION
%%

\section{Linear algebra lemmas} 

This section presents a series of lemmas that we rely on to build circuits that prepare 3-qubit states. Let us start with some definitions.

\begin{mydef} The controlled-$Z$ gate acting on two qubit states is represented by the matrix 

$$cz=\left(
\begin{array}{cccc}
 1 & 0 & 0 & 0 \\
 0 & 1 & 0 & 0 \\
 0 & 0 & 1 & 0 \\
 0 & 0 & 0 & -1 \\
\end{array}
\right)$$

For three qubits  we will denote by $cz_{ij}$ (with $i\ne j$) the controlled-$Z$ gate acting on  qubits $i$ and $j$, where the qubits are labelled  $0$,  $1$ and  $2$ from right to left. For example,  $\ket{011}$ means that qubit $0$ is in  state $\ket{1}$, qubit $1$ in in  state $\ket{1}$ and qubit 2 is in  state $\ket{0}$. Using $e_1=(1,0,0,\dots, 0), e_2=(0,1,0,\dots, 0)$, etc. as colmuns, we have  $cz_{01}=[e_1,e_2,e_3,-e_4,e_5,e_6,e_7,-e_8]$, $cz_{02}=[e_1,e_2,e_3,e_4,e_5,-e_6,e_7-e_8]$ and $cz_{12}=[e_1,e_2,e_3,e_4,e_5,e_6,-e_7-e_8]$.
\end{mydef}

\begin{mydef}

(i) For any 3-qubit state

%$$\ket{\phi}=t_{000}\ket{000}+t_{001}\ket{001}+t_{010}\ket{010}+t_{011}\ket{011}+ t_{100}\ket{100}+t_{101}\ket{101}+t_{110}\ket{110}+t_{111}\ket{111}\, ,$$

$$\ket{\phi}=t_{000}\ket{000}+t_{001}\ket{001}+\dots+t_{110}\ket{110}+t_{111}\ket{111}\, ,$$

we define  $2\times 2$ matrices $T_0$ and $T_1$ by

$$T_0=\begin{pmatrix} t_{000}&t_{001}\\t_{010}&t_{011}\end{pmatrix}\quad \hbox{and}\quad T_1=\begin{pmatrix} t_{100}&t_{101}\\t_{110}&t_{111}\end{pmatrix}$$

and we write $\ket{\phi}=\ket{0}T_0+\ket{1}T_1$.

(ii) For any 2-qubit state $\ket{\phi}=t_{00}\ket{00}+t_{01}\ket{01}+t_{10}\ket{10}+t_{11}\ket{11}$, we define the matrix $T(\ket{\phi})=\begin{pmatrix}t_{00}&t_{01}\\t_{10}&t_{11}\end{pmatrix}$

\end{mydef}

\begin{mydef} \label{m} (i) For any pair of complex numbers $(x,y)\ne(0,0)$, we define the unitary matrix

$$U(x,y)=\frac{1}{\sqrt{|x|^2+|y|^2}}\, \begin{pmatrix} x & y \\ -\bar{y} & \bar{x} \end{pmatrix} $$ 

where $\bar{z}$ denotes the conjugate of the complex number $z$. 

(ii) For any nonsingular matrix $A= \begin{pmatrix}a&b\\c&d\end{pmatrix}$ we define the  unitary matrix $R_1(A)=U(x,y)$ where

	$$\begin{cases} x=d-bk\\ y=\bar{c}-\bar{a} \bar{k}\end{cases} \hbox{with}\quad k=\begin{cases} \sqrt{\frac{|c|^2+|d|^2}{|a|^2+|b|^2}}                      &\,   \hbox{ if }\,  \beta= a\bar{c}+b\bar{d} =0 \\
	     -  \sqrt{\frac{|c|^2+|d|^2}{|\beta|^2(|a|^2+|b|^2)}}    \, \bar{\beta}&  \, \hbox{ if } \, \beta= a\bar{c}+b\bar{d} \ne 0 
	        \end{cases}$$

	 (iii) For a nonzero matrix  $A= \begin{pmatrix}a&b\\c&d\end{pmatrix}$  with  determinant zero, we define $R_2(A)=U(x,y)$ where either

	$$\begin{cases} x=b\\ y=\bar{a}- \bar{k} \end{cases} \hbox{with}\quad    k=\begin{cases} \sqrt{|a|^2+|b|^2}                      &\,   \hbox{ if }\,  a=0 \\
	     -  \sqrt{\frac{|a|^2+|b|^2}{|a|^2}}    \, a&  \, \hbox{ if } \, a \ne 0 
	        \end{cases}$$
	        
if the first row of $A$ is not the zero vector or 

		$$\begin{cases} x=d\\ y=\bar{c}- \bar{k} \end{cases} \hbox{with}\quad    k=\begin{cases} \sqrt{|c|^2+|d|^2}                      &\,   \hbox{ if }\,  c=0 \\
	     -  \sqrt{\frac{|c|^2+|d|^2}{|c|^2}}    \, a&  \, \hbox{ if } \, c \ne 0 
	        \end{cases}$$
		
if the first row of $A$ is the zero vector. 	 

(iv) For a nonzero matrix  $A= \begin{pmatrix}a&b\\c&d\end{pmatrix}$ with determinant zero, we choose $(v_1,v_2)$ such that both columns of $A$  are multiples of $(v_1,v_2)$, and define   $L_1(A)=U(\bar{v_1},\bar{v_2})$.

(v) For a nonzero matrix  $A= \begin{pmatrix}a&b\\0&0\end{pmatrix}$  we define $R_3(A)=U(\bar{a},-b)$.

\end{mydef}

 \begin{rem} If the entries of matrix $A$ in the previous lemma are real numbers, then $U(x,y)$, $R_i(A)$ and $L_1(A)$ either have real entries or can be chosen to have real entries.
 \end{rem}

The following lemma shows how matrices $T_0$ and $T_1$  change when we apply either a local gate or a controlled-$Z$ gate to the state $\ket{\phi}=\ket{0} T_0+\ket{1}T_1$. The proof is a direct computation.

\begin{lem}\label{p3q} Let  $U=\begin{pmatrix} u_{00}&u_{01}\\u_{10}&u_{11}\end{pmatrix}$ be a unitary matrix. Let $I$ denote the $2\times 2$ identity matrix and $Z=\begin{pmatrix}1&0\\0&-1  \end{pmatrix}$. If \,$\ket{\phi}=\ket{0}T_0+\ket{1}T_1$, then 

\begin{itemize}
\item $U\otimes I\otimes I\,\ket{\phi}= \ket{0}(u_{00}T_0+u_{01}T_1)+\ket{1}(u_{10}T_0+u_{11}T_1)$
\vskip.2cm
\item $I\otimes U\otimes I\,\ket{\phi}= \ket{0}(UT_0)+\ket{1}(UT_1)$
\vskip.2cm
\item $I\otimes I\otimes U\,\ket{\phi}= \ket{0}(T_0U^T)+\ket{1}(T_1U^T)$
\vskip.2cm

\item $cz_{01} \ket{\phi}= \ket{0}\begin{pmatrix} t_{000}&t_{001}\\t_{010}&-t_{011}\end{pmatrix}+\ket{1}\begin{pmatrix} t_{100}&t_{101}\\t_{110}&-t_{111}\end{pmatrix}$

\vskip.2cm
\item $cz_{02} \ket{\phi}= \ket{0}T_0+\ket{1}T_1Z$

\vskip.2cm
\item $cz_{12} \ket{\phi}= \ket{0}T_0+\ket{1}ZT_1$
\end{itemize}

\end{lem}

The following lemma describes the form of matrices $T_0$ and $T_1$ when the state $\ket{\phi}=\ket{0} T_0+\ket{1}T_1$ can be written as the tensor product of two states.

\begin{lem}   \label{rightfactor} 
A state $\ket{\phi}=\ket{0}T_0+\ket{1}T_1$  can be written as 

$$(b_0\ket{00}+b_1\ket{01}+b_2\ket{10}+b_3\ket{11})\otimes (a_0\ket{0}+a_1\ket{1})$$

 if and only if the rows of $T_0$ and $T_1$ are all multiples of the vector $(a_0,a_1)$.
\end{lem}

Once again, the proof is straightforward.

\begin{lem}  \label{proportionalrows} 
	Let $   \dd{ A  = \matrixtwo{a}{b}{c}{d}} $ be a nonsingular complex matrix.  With $R_1(A)$ and $k$ as given  in Definition \ref{m},  the matrix $  W  = \begin{pmatrix} w_{11} &w_{12} \\ w_{21} & w_{22} \end{pmatrix}  =  A R_1(A) $  satisfies 
		\begin{equation}  \label{condition}  w_{21} = k w_{11}    \qquad  \text{and} \qquad  w_{22}  =  -kw_{12} \end{equation}
\end{lem}

\begin{proof}  Since $det(A)\ne 0$ then $(x,y)\ne(0,0)$.    A direct computation shows that Equation \eqref{condition} can be written in matrix form in the form $  C \columntwo{x}{\bar{y}}  =  \columntwo{0}{0}$ where $ C $ is the matrix  
	
	\[  C =  \matrixtwo{c - ka}{-d + kb}{\bar{d} + \bar{k} \bar{b}}{\bar{c} + \bar{k}{\bar{a}}}.   \]
	
Observe that the determinant of $C$ can be written in the form 

\begin{equation} \label{deterC}  \deter{C}  =  \alpha |k|^2  + \beta k  +  \gamma \bar{k}  +  \delta,  \end{equation}  

where 

\[  \alpha  =  - ( |a|^2 + |b|^2  ),  \quad \beta =  - ( a \bar{c} + b \bar{d}),    \quad  \gamma =   \bar{a}c + \bar{b} d,   \quad \delta  =   |c|^2 + |d|^2    \]

We see that    $  \gamma   =  - \bar{\beta}, $ \  hence  

\[ \beta k  +  \gamma \bar{k}  =  \beta k  -  \overline{\beta k} =    - 2 \im{(\beta  k)}  \cdot   i \]   

which is a purely imaginary number. It follows that our choice for $k$ in the statement of the lemma makes the determinant of $C$ vanish. Therefore the equation $C \columntwo{x}{\bar{y}}  =  \columntwo{0}{0}$ has infinitely many solutions. Clearly our choice $x$ and $y$ is one of these solutions.

   \end{proof}

 \begin{lem}   \label{twomatrices}
 Let   $ \dd{  A  =  \matrixtwo{\alpha}{0}{0}{0} }$ and $ \dd{ B  =  \matrixtwo{b_{11}}{b_{12}}{b_{21}}{b_{22}} }$  be nonzero matrices  with  determinant zero. If $R_2(B)$ is defined as in Definition \ref{m} then all four rows of  the matrices $ A R_2(B) $ and $ B R_2(B) Z $ are multiples of a single row.
        
 \end{lem}

\begin{proof}   Let us assume that the first row of $B$ is not the zero vector. Since $\det(B)= 0$, the two rows of  matrix $ B U(x,y) Z $ are multiples of each other. Therefore it is enough to show that  our choice of $U(x,y)$ makes the first row of $BU(x,y)Z$ a multiple of the first row of $AU(x,y)$. The first row of  $AU(x,y)$ is a multiple of the vector $(x,y)$ and the first row of $BU(x,y)Z$ is a multiple of vector $(b_{11} x-b_{12}\bar{y},-b_{11}y-b_{12} \bar{x})$. The equation  $(b_{11} x-b_{12}\bar{y},-b_{11}y-b_{12} \bar{x})=k(x,y)$ reduces to the equation  $ C \columntwo{x}{\bar{y}}  =  \columntwo{0}{0}$ where $ C $ is the matrix  
	
	\[  C =  \matrixtwo{b_{11} - k}{-b_{12}}{\bar{b_{12} }}{\bar{b_{11}} + \bar{k}}.   \]
	
Observe that the determinant of $C$ can be written in the form 

\begin{equation*} \label{deterC}  \deter{(C)}  =  -|k|^2  -  \bar{b}_{11} k  +  b_{11} \bar{k}  +  |b_{11}|^2 + |b_{12}|^2 =|b_{11}|^2 + |b_{12}|^2-|k|^2-2i \im{(k\bar{b}_{11})},  \end{equation*}  

If follows that our choice for $k$ in the statement of the lemma makes the determinant of $C$ vanish. Therefore the equation $C \columntwo{x}{\bar{y}}  =  \columntwo{0}{0}$ has infinitely many solutions. Clearly our choice $x$ and $y$ is one of these solutions. The case when the first row of $B$ is the zero  vector is similar.

   \end{proof}
   
\section{The circuits}

In this section we  explicitly show a circuit that takes any 3-qubit state into  $\ket{000}$ using local gates and at most three controlled-$Z$ gates. We  also explicitly show a circuit that takes any real 3-qubit state into $\ket{000}$ using real local gates and at most four controlled-$Z$ gates.

\begin{rem}\label{iba} Here is an overview of the construction of the circuit that takes a given 3-qubit state into the state $\ket{000}$:
\begin{enumerate}
\item We start with a 3-qubit state $\ket{\varphi_0}=\ket{0} A_0+\ket{1} B_0$ where $A_0$ and $B_0$ are $2\times 2$ matrices.
\item We use a local gate to change $\ket{\varphi_0}$ to $\ket{\varphi_1}=\ket{0} A_1+\ket{1} B_1$ with $\det{(A_1)}=0$. This step uses a trick from \cite{AA}.
\item We use local gates and a controlled-$Z$ gate (if needed) to change  $\ket{\varphi_1}$ to $\ket{\varphi_2}=\ket{0} A_2+\ket{1} B_2$ with  $\det{(A_2)}=0$ and  $\det{(B_2)}=0$. This uses Lemma \ref{proportionalrows}.
\item We use local gates and a controlled-$Z$ gate (if needed) to change $\ket{\varphi_2}$ to $\ket{\varphi_3}=\ket{0} A_3+\ket{1} B_3$ with all the rows of the two matrices $A_3$ and $B_3$ multiples of each other. This uses Lemma \ref{twomatrices}. 
\item By Lemma \ref{rightfactor} we see that 3-qubit state $\ket{\varphi_3}$ is the tensor product of a 2-qubit state and a 1-qubit state. We use a controlled-$Z$ gate (if needed) to unentangle the 2-qubit state and reach the state $\ket{000}$.

\end{enumerate}
\end{rem}

Before fully describing the preparation of 3-qubit states, it is convenient to show how the ideas in this paper can be used to prepare any 2-qubit state.

%
% Preparing two qubit states
%
\begin{prop} \label{p2q}Let $\ket{\varphi_0}=\xi_0\ket{00}+\xi_1\ket{01}+\xi_2\ket{10}+\xi_3\ket{11}$ be a 2-qubit state. Denoting $A=\begin{pmatrix}\xi_0&\xi_1\\ \xi_2&\xi_3\end{pmatrix}$ we have that
\begin{enumerate}
\item
If $\det{(A)}=0$, then $(K_1\otimes K_2)  \ket{\varphi_0}=\ket{00}$ where  $K_1=L_1(A)$ with $L_1(A)$ defined in Definition \ref{m}, and $K_2=U(\bar{\eta_0},-\eta_1)^T$ where $\eta_0\ket{00}+\eta_1\ket{01}=(K_1\otimes I) \ket{\varphi_0}$.
\item

If $\det{(A)}\ne 0$, then $(K_1\otimes K_2)cz(I\otimes W_1) \ket{\varphi_0}=\ket{00}$ where $W_1=R_1(A)$ where $R_1(A)$ is defined as in the Definition \ref{m}
and the matrices $K_1$ and $K_2$ are computed as in part (1) using the 2-qubit $cz(I\otimes W_1)\ket{\varphi_0}$.

\end{enumerate}
\end{prop}

\begin{proof} Let us prove part (1), since $\det{(A)}=0$, then the two columns of $A$ are multiple of each other. We can check that the $(K_1\otimes I) \ket{\varphi_0}=b_{11}\ket{00}+b_{12}\ket{01}+b_{21}\ket{10}+b_{22}\ket{11}$ with the matrix $B=\{b_{ij}\}$ satisfying $B=K_1A$. Since the conjugate of the  first row of $K_1$ is a multiple of both rows of $A$ then, the second row of $B$ vanishes. We can check if  $(K_1\otimes K_2) \ket{\varphi_0}=\gamma_0\ket{00}+\gamma_1\ket{01}+\gamma_2\ket{10}+\gamma_3\ket{11}$ then the matrix 
$C=\begin{pmatrix}\gamma_0&\gamma_1\\ \gamma_2&\gamma_3\end{pmatrix}$ satisfies $C=BK_2^T$ which is equal to the matrix $\begin{pmatrix}1&0\\ 0&0\end{pmatrix}$ because the first column of $K_2^T$ is equal to the conjugate of the first row of $B$. Recall that the second row of $B$ vanishes. To prove part (2)  we notice that if $(I\otimes W_1) \ket{\varphi_0}=b_{11}\ket{00}+b_{12}\ket{01}+b_{21}\ket{10}+b_{22}\ket{11}$, then $B=AW_1^T$. From  Lemma \ref{proportionalrows} we obtain that the two rows fo $cz(I\otimes W_1) \ket{\varphi_0}$ are multiple of each other. Once we know that $\ket{\varphi_1}=cz(I\otimes W_1) \ket{\varphi_0}=c_{11}\ket{00}+c_{12}\ket{01}+c_{21}\ket{10}+c_{22}\ket{11}$ with $\det{(C)}=0$, then the rest of the proof follows from part (1) in this proposition.

\end{proof}

%\begin{rem} The two cases in the previous proposition correspond to the cases $
%\end{rem}

%\begin{pmatrix}\xi_0&\xi_1\\ \xi_2&\xi_3\end{pmatrix}
%%%
%%%
%%%
\subsection{The algorithm that prepares any 3-qubit state} \label{a3q} In this section we  developed  the ideas presented in Remark \ref{iba}.

\begin{enumerate}
\item
Let us start with the 3-qubit state $\ket{\varphi_0}=\ket{0} A_0+\ket{1} B_0$. If $\det{(B_0)}\ne0$ let us define $U_1=U(1,z_0)$ where $z_0$ is a solution of the quadratic equation $\det{(A_0+zB_0)}=0$. If $\det{(B_0)}=0$ then we take $W_1=\begin{pmatrix}0&1\\-1&0\end{pmatrix}$. Using Lemma \ref{p3q} we have that  $\ket{\varphi_1}=(W_1\otimes I\otimes I)\ket{\varphi_0}=\ket{0} A_1+\ket{1} B_1$ satisfies that $\det{(A_1)}=0$. If $A_1$ is the zero matrix then $\ket{\varphi_1}=\ket{1}\otimes \ket{\phi_0}$ and then, $(X\otimes L)\ket{\varphi_1}=\ket{00}$ where $X=\begin{pmatrix}0&1\\ 1&0\end{pmatrix}$ and $L=K_1\otimes K_2$ or $L=(K_1\otimes K_2)cz(I\otimes W_1) $ is the matrix given by Proposition \ref{p2q} for the 2-qubit state $ \ket{\phi_0}$. Therefore, in this case,  the state $\ket{\varphi_0}$ can be prepare with at most one controlled-$Z$ gate. We continue the algorithm  assuming that $A_1$ is not the zero matrix.

\item
 We define $U_2=L_1(A_1)$ and  $\ket{\varphi_2}=(I\otimes U_2\otimes I)\ket{\varphi_1}=\ket{0} A_2+\ket{1} B_1$. We can check that the second row of $A_2$ vanishes.
 \item
 
We define $U_3=R_3(A_2)^T$ and  $\ket{\varphi_3}=(I\otimes I\otimes U_3)\ket{\varphi_2}=\ket{0} A_3+\ket{1} B_3$. We can check that the second row  and the second column of $A_3$  vanish.

\item 

We now assume that $\det{(B_3)}\ne0$, if $\det{(B_3)}=0$ then we can skip this step and the qubit $\varphi_0$ can be prepared with at most two controlled-$Z$ gates. 

We define $U_4=R_1(B_3)^T$ and  $\ket{\varphi_4}=(I\otimes I\otimes U_4^{-1})cz_{01}(I\otimes I\otimes U_4)\ket{\varphi_3}=\ket{0} A_4+\ket{1} B_4$. We can check that the second row   
and the second column of $A_4$  vanish and also, $\det{(B_4)}=0$.

\item 

We now assume that the second column of $B_4$ does not vanish, if it does,  then we can skip this step and the qubit $\varphi_0$ can be prepared with at most two controlled-$Z$ gates. 
We define  $U_5=R_1(B_4)^T$ and  

\begin{eqnarray}
\ket{\varphi_5}&=&cz_{02}(I\otimes I\otimes U_5)\ket{\varphi_4}=\ket{0} A_5+\ket{1} B_5\\
& =& \xi_0\ket{000}+\xi_1\ket{001}+\dots+\xi_7\ket{111}
\end{eqnarray}

 We can check that the rows of  the matrices $A_5$ and $B_5$ are multiple of each other. Using Lemma \ref{rightfactor} to deduce that $\ket{\varphi_5}=\ket{\phi_0}\otimes \ket{\phi_1}$. We can check that we can pick  ${\phi_1}=a_1\ket{0}+a_2\ket{1}$ with $(a_1,a_2)$ the first row of $A_5$ normalized and, if $a_1\ne0$,  $\ket{\phi_0}=b_1\ket{00}+b_2\ket{10}+b_3\ket{11}$  with 

$$b_1=\frac{\xi_0}{a_1},\quad b_2=\frac{\xi_4}{a_1}, \quad b_3=\frac{\xi_6}{a_1}$$, 

If $b_3=0$, then $\ket{\varphi_5}=\ket{\phi_0}\otimes \ket{\phi_1}=(b_1\ket{0}+b_2\ket{1})\otimes \ket{0}\otimes \ket{\phi_0}$ and therefore $K_1\otimes I\otimes K_2\ket{\varphi_5}=\ket{000}$ with  $K_1=U(\bar{a_1},-a_2)^T$ and $K_2=U(\bar{b_1},-b_2)^T$. If $b_3\ne0$ then by  Proposition \ref{p2q} we can find gates $U_6$, $U_7$ and $U_8$ such that $(U_7\otimes U_8)cz(I\otimes U_6)\ket{\phi_0}=\ket{00} $, therefore, if $U_9=U(\bar{a_1},-a_2)^T$ then

$$(U_9\otimes U_7\otimes U_8)cz_{12}(I\otimes I\otimes U_6)\ket{\varphi_5}=\ket{000}$$

\end{enumerate}

\subsection{The algorithm that  prepares any real 3-qubit state} Let us consider any real 3-qubit state $\phi_0=\ket{0}A_0+\ket{1}B_0$. Notice that if $\det(A_0)=0$ or $\det(B_0)=0$ or both determinants  $\det(A_0)$ and $\det(B_0)$ are different from zero and the quadratic equation $\det(A_0+zB_0)=0$ has a real solution, then following the steps in subsection \ref{a3q} give us way to prepare $\ket{\phi_0}$ using real local gates and at most three controlled-$Z$ gates. Notice that if 

$$\ket{\phi_0}=w_0\ket{000}+w_1\ket{001}+w_2\ket{010}+w_3\ket{011}+w_4\ket{100}+w_5\ket{101}+w_6\ket{110}+w_7\ket{111}$$

then, the three cases  (i) $\det(A_0)=0$, (ii) $\det(B_0)=0$ and (iii) both determinants  $\det(A_0)$ and $\det(B_0)$ are different from zero and the quadratic equation $\det(A_0+zB_0)=0$ has a real solution, imply that

$$\Delta(\ket{\phi})=(w_0 w_7-w_1 w_6-w_2 w_5+w_3 w_4)^2-4 (w_1 w_2-w_0 w_3) (w_5 w_6-w_4 w_7)\ge 0$$

Therefore, we conclude that if $\Delta(\ket{\phi})\ge 0$, then $\ket{\phi_0}$ can be prepared with real local gates and at most three controlled-$Z$ gates. We can show that if $\Delta(\ket{\phi})
<0$, then the gate $U_0=R_1(A_0)$ is a real gate and $\ket{\phi_1}=cz_{01}\, (I\otimes I\otimes U_0)\ket{\phi_0}$ can be writes as $\ket{0}A_1+\ket{1}B_1$ with $\det(A_1)=0$. Therefore, applying the procedure in subsection \ref{a3q} we have that $\ket{\phi_1}$ can be taken into the state $\ket{000}$ by using real local gates and at most three controlled-$Z$ gates.

\section{Conclusions}

This paper shows an explicit and optimal way to prepare any 3-qubit state different from the one presented in \cite{Z} by Znidaric, Giraud and Georgeot. It also shows a way to prepare any real 3-qubit state using real local gates and at most 4 controlled-$Z$ gates. Using only real local gates, the question if four  is the least amount of controlled-$Z$ gates needed to prepare any real 3-qubit state remains open. The method presented in this paper heavily use the fact we can write 3-qubit states using two 2 by 2 matrices. This idea was taken from the paper \cite{AA}. All the steps in the algorithm can be view as steps with the purpose of gaining some properties for these two matrices.

\end{document}